\documentclass[11pt]{gen-j-l}

\usepackage[hmargin=1in,vmargin=1in]{geometry}
\usepackage{amsmath,amsthm,amssymb,amsfonts,verbatim,mathptmx}

\usepackage[dvipsnames,usenames]{color}
\usepackage{hyperref}

\usepackage{todonotes}
\DeclareMathAlphabet{\mathcal}{OMS}{cmsy}{b}{n}
\DeclareMathAlphabet{\mathcal}{OMS}{cmsy}{m}{n}

\title[ New extension constructions of optimal FHS sets]{\large New extension constructions of optimal frequency hopping sequence sets}

\email{rurustef1212@gmail.com}
\thanks{This work is supported by the National Science Foundation
of China(Grant No. 61401369), the Youth Science and Technology
Fund of Sichuan Province(No.2017JQ0059) }
\newtheorem{Theorem}{Theorem}

\newtheorem{Corollary}{Corollary}

\newtheorem{Example}{Example}

\newtheorem{Remark}{Remark}
\newtheorem{Lemma}{Lemma}

\newtheorem{Construction}{Construction}

\def\ZZ{\mathbb{Z}}

\def\FF{\mathbb{F}}

\def \mF {\mathcal{F}}
\def \mE {\mathcal{E}}

\def \mS {\mathcal{S}}
\def \mU {\mathcal{U}}
\def \mX {\mathcal{X}}
\def \mY {\mathcal{Y}}

\def\lpf {{\rm lpf}}

\def \bx {{\bf x}}
\def\be {{\bf e}}
\def \by {{\bf y}}
\def \bs {{\bf s}}

\def \bu {{\bf u}}

\def\ord{{\rm ord}}

\begin{document}
\maketitle
\thispagestyle{empty}

\centerline{\scshape Xianhua Niu}
{\footnotesize
% please put the address of the first author
 \centerline{School of Computer and Software Engineering}
  \centerline{Xihua University}
   \centerline{National Key Laboratory of Science and Technology on Communications}
    \centerline{University of Electronic Science and Technology of China}
     \centerline{Chengdu, China.}
} % Do not forget to end the {\footnotesize by the sign }

\medskip
\centerline{\scshape Chaoping Xing}

{\footnotesize
 \centerline{Division of Mathematical Sciences}
  \centerline{School of Physical and Mathematical Sciences}
   \centerline{Nanyang Technological University}
    \centerline{Singapore.}
} % Do not forget to end the {\footnotesize by the sign }

\begin{abstract}
In this paper, a general framework of constructing optimal frequency hopping sequence (FHS) sets is presented based on the designated direct product. Under the framework, we obtain infinitely many new optimal FHS sets by combining a family of sequences that are newly constructed in this paper with some known optimal FHS sets. Our constructions of optimal FHS sets are also based on extension method. However, our constructions remove the constraint requiring that the extension factor is co-prime with the length of original FHSs and get new parameters. In literature, most of the extension constructions suffer from this constraint. As a result, our constructions allow a great flexibility of choosing parameters of FHS sets for a given frequency-hopping spread spectrum system.
\end{abstract}

\section{Introduction}

Frequency hopping (FH) multiple-access (MA) spread spectrum systems, with its anti-jamming, secure and multiple access properties, have
found many applications in military radio communications, mobile communications, modern radar and sonar echolocation systems \cite{Fan,Golomb}.
In such systems, multiple frequency-shift keying (MFSK) is the data modulation technique employed, and wideband signals are generated by hopping from one frequency slot to anther over a large number of frequency slots. The frequency slots used are chosen pseudo-randomly by a code called frequency hopping sequence (FHS). As is often the case, in a multiple access environment, an important requirement is to keep the mutual interference between transmitters on a level as low as possible. This mutual interference occurs when two or more sources transmit the same frequency slot at the same time. The degree of the mutual interference is clearly related to the Hamming correlation properties of the frequency hopping sequences \cite{Fan,Simon}. In order to improve their performance, it is desirable to employ frequency hopping sequences (FHSs) having low Hamming correlation to reduce the multiple-access interference (also called hits) of frequencies \cite{LG}. Thus, the design of an FHS set with good property is an important problem.

The main purpose of FHS design is to find an FHS or an FHS set which is optimal under a given condition. In general, the optimality of an FHS set is measured by the Peng-Fan bound \cite{PF}, whereas that of a single FHS is by the Lempel-Greenberger bound \cite{LG}. There are several algebraic or combinatorial constructions for optimal FHSs or FHS sets in the literature \cite{KUMAR}-\cite{XU-16}. Moreover, some extension methods have been proposed \cite{Chung-09}-\cite{Xu-16}, which generate several new families of optimal FHSs. By applying some known FHS sets, the existing extension constructions have produced new FHS sets with some desirable parameters. However, most of the extension constructions suffer from the constraint that the extension factor is co-prime with the length of original FHS sets.

In this paper, we present a general framework of constructing optimal FHS sets based on the designated direct product. Under the framework, we obtain infinitely many new optimal FHS sets which increase the length and alphabet size of the original FHS set by using flexible extension factor, but preserve its maximum Hamming correlation, as summarized in Table \ref{tab:1}. Moreover, new constructions remove the constraint requiring that the extension factor is co-prime with the length of original FHSs. As a result, we have a great flexibility of choosing parameters of FHS sets for a given frequency-hopping spread spectrum system.

\begin{table}[tbp]\label{tab:1}\small
\newcommand{\tabincell}[2]{\begin{tabular}{@{}#1@{}}#2\end{tabular}}
\centering  % 表居中
\caption{The Extended FHS Set From a $(N, v, \lambda; M)$ -FHS Set $\mX$. }

\begin{tabular}{l|l|l|l}  %p{3cm} {lccc} 表示各列元素对齐方式，left-l,right-r,center-c
\hline
 Extended FHS set & Constraints &References &Remarks \\ \hline  % \hline 在此行下面画一横线

$(nN,nv,\lambda;M)$ &\tabincell{l}{$m(\mX)\leq \lpf(n)-1$, $\gcd(n,N)=1$.} &\cite{Chung-14}\cite{Xu-16} & \\ \hline
$(nN,nv,\lambda;M)$ &\tabincell{l}{$m(\mX)\leq \lpf(n)-1$. } &\footnotesize{Corollary 1}  &{\footnotesize{co-prime condition is removed}}\\ \hline \hline
$((q-1)N,qv,\lambda;M)$ &\tabincell{l}{$m(\mX)\leq q$, $\gcd((q-1),N)=1$.} &\cite{Chung-14}\cite{Xu-16} & \\ \hline
$((q-1)N,qv,\lambda;M)$ &\tabincell{l}{$m(\mX)\leq q$.} &\footnotesize{Corollary 3} &{\footnotesize{co-prime condition is removed}}\\ \hline
$(dN,qv,\lambda;M)$ &\tabincell{l}{$m(\mX)\leq p$, $d=(p-1)p^{a-1}$.} &\footnotesize{Theorem 2}  &\footnotesize{different parameter regime}\\ \hline \hline
$(n(q-1)N,nqv,\lambda;M)$ &\footnotesize{\tabincell{l}{$m(\mX)\leq \min\{\lpf(n)-1,q\}$, $\gcd(n,N)=1$,\\$\gcd((q-1),N)=\gcd((q-1),n)=1$.} } &\cite{Chung-14}\cite{Xu-16}  & \\ \hline
$(n(q-1)N,nqv,\lambda;M)$ &\tabincell{l}{$m(\mX)\leq \min\{\lpf(n)-1,q\}$.} &\footnotesize{Corollary 5} &{\footnotesize{co-prime condition is removed}}\\ \hline
$(ndN,nqv,\lambda;M)$ &\tabincell{l}{$m(\mX)\leq \min\{\lpf(n)-1,p\}$, \\$d=(p-1)p^{a-1}$.} &\footnotesize{Corollary 4} &{\footnotesize{different parameter regime}}\\ \hline
\end{tabular}
%\medskip
\footnotesize In  Table \ref{tab:1}, $\lpf(n)$ denotes the least prime factor of $n$, $m(\mX)$ denotes the maximum number of appearance of frequency slot in FHS set $\mX$ and $q=p^a$ for a prime $p$.
\end{table}
The rest of this paper is organized as follows. In Section 2, we give some preliminaries to FHSs. In Section 3, we present a general framework to construct optimal FHS sets based on the designated direct product. In Section 4, we obtain the new construction of optimal FHS sets with length $nN$. In Section 5, we give a construction of optimal FHS sets with new parameter of length $dN$ or $ndN$. Finally, we conclude the paper in Section  6.

\section{Preliminaries}

Throughout this paper, the following notations will be used:

$p_1, \ldots, p_r:$  primes with $p_1 <\ldots< p_r$;

$q_1:$ power of prime $p_1$;

$q:$ power of prime $p$;

$\langle x\rangle_y:$ the least nonnegative residue of $x$ modulo $y$ for an integer $x $ and a positive integer $y$;

$\ZZ_n:$ the ring of integers modulo $n$ for a positive integer $n>1$;

$\ord(g):$ the multiplicative order of $g \in \ZZ_{n}^*$;

$\lceil z\rceil:$ the largest integer less than or equal to $z$;

$(N,v,\lambda):$ an FHS of length $N$ over a frequency slot set of size $v$, with the maximum Hamming autocorrelation $\lambda$;

$(N,v,\lambda;M):$ an FHS set of $M$ sequences of length $N$ over a frequency slot set of size $v$, with the maximum Hamming correlation $\lambda$.

Let $\mF=\{f_1,f_2,\ldots,f_v\}$ be a frequency slot set with size $|\mF|=v$, $\mX$ be a set of $M$ FHSs of
length $N$. For any two FHSs $\bx_i=\left(x_i(0),x_i(1),\ldots,x_i(N-1)\right)$, $\bx_j=\left(x_j(0),x_j(1),\ldots,x_j(N-1)\right)\in \mX$, $0\leq i\neq j\leq M-1$, the Hamming correlation function $H_{\bx_i\bx_j}(\tau)$ of sequences $\bx_i$ and $\bx_j$ at time delay $\tau$ is defined as follows:
\begin{equation}\label{e1}
H_{\bx_i\bx_j}(\tau)=\sum^{N-1}_{t=0}h\left(x_i(t),x_j(\langle t+\tau\rangle_N)\right), ~~0\leq\tau\leq N-1,
\end{equation}
where $h(a, b)=1$ if $a=b$, and $h(a, b)=0$ otherwise. And only positive time shifts are considered.

For any given FHS set $\mX$, the maximum Hamming autocorrelation $H_a(\mX)$, the maximum Hamming crosscorrelation $H_c(\mX)$ and the maximum Hamming correlation $H_m(\mX)$ are defined as follows, respectively:
\begin{eqnarray*}
  H_a(\mX)&=&\max\limits_{0<\tau<N}\left\{H_{\bx_i\bx_i}(\tau):~\bx\in \mX\right\},\\
  H_c(\mX)&=&\max\limits_{0\leq \tau <N}\left\{H_{\bx_i\bx_j}(\tau):~\bx_i,\bx_j\in \mX, i\neq j\right\}.\\
  H_{m}(\mX)&=&\max\left\{H_a(\mX),H_c(\mX)\right\}.
\end{eqnarray*}

In 2004, Peng and Fan \cite{PF} established the following bound of an FHS set.

\begin{Lemma}[Peng-Fan bound]\label{PF-bound}
 Let $\mX$ be an FHS set of $M$ sequences and length $N$ over a given frequency slot set $\mF$ of size $v$, we have
\begin{equation}\label{e2}
  H_{m}(\mX)\geq \left\lceil\frac{(MN-v)N}{(MN-1)v}\right\rceil.
\end{equation}
\end{Lemma}

An $(N,v,\lambda;M)$ FHS set $\mX$ is called {\it optimal} if the Peng-Fan bound in Lemma \ref{e2} is met with equality. An $(N,v,\lambda)$ FHS $\bx$ is called {\it optimal} if the Lempel-Greenberger bound is met with equality.

In practical applications, the required length and alphabet size of an FHS or an FHS set are variable according to the specification of a given system or environment. Thus, it is very important to select optimal FHSs or FHS sets with flexible parameters under the given condition.

\begin{Lemma}
 Let $\mX=\left\{\bx_i=(x_i(0),x_i(1),\ldots,x_i(N-1)):0\leq i\leq M-1\right\}$ be an FHS set over frequency slot set $\mF=\{f_1,f_2,\ldots,f_v\}$. For any $f_k\in \mF$, let
 \begin{equation*}
\Phi_k=\{(i,a)| x_i(a)=f_k\}, 1\leq k \leq v.
\end{equation*}

Then, the maximum number of appearance of any frequency slot $f_k \in \mF$ in FHS set $\mX$, denoted by  $m(\mX)$, can be written as $m(\mX)=\max\limits_{1\leq k\leq v}\left\{\mid \Phi_k\mid:f_k\in \mF\right\}$.
\end{Lemma}

In the following section, a framework based on the designated direct product will be given, which can be used to construct optimal FHS sets with new parameters.

\section{A General Framework of Extension construction of FHS set}
In this section, we give a framework of extension construction of FHS set based on the designated direct product by combining a family of sequences with some known optimal FHS sets.

For a $(N,v,\lambda; M)$ FHS set $\mX=\{\bx_i:0\leq i\leq M-1\}$ over $\mF$, with
\[
\bx_i=(x_i(0),x_i(1),\ldots,x_i(N-1)).
\]

Let $\mE=\{\be_\delta:0\leq \delta\leq l-1\}$ be a sequence set of length $n$ over $\ZZ_m$ with
\[
\be_\delta=(e_\delta(0),e_\delta(1),\ldots,e_\delta(n-1)).
\]

If $m(\mX)\leq l$, let $\omega_i(t_2)$ be a function with $\omega_i(t_2)=\varphi_{f_k}(i,t_2)$, where $\varphi_{f_k}$ is a injective function from $\Phi_k$ to $Z_{l}^*$ for every $k$ with $0\leq k \leq v$, $0\leq i\leq M-1$, $0\leq t_2\leq N-1$, the set $\Phi_k$ is defined as Lemma 2.

Then, an $n\times N$ matrix is formed by combining the sequence $\be_{\omega_{i}(t_2)}$ with the FHS $\bx_{i}$ as follows:
\begin{eqnarray}\label{e4}\small
\ \ \ \bu_{i}\!\!&\!\!=\!\!&\!\!\left( {\begin{array}{*{20}{c}}
   {\left(e_{\omega_{i}(0)}(0),x_{i}(0)\right)} & {\left(e_{\omega_{i}(1)}(0),x_{i}(1)\right)} &  \ldots  & {\left(e_{\omega_{i}(N-1)}(0),x_{i}(N-1)\right)}  \\
   {\left(e_{\omega_{i}(0)}(1),x_{i}(0)\right)} & {\left(e_{\omega_{i}(1)}(1),x_{i}(1)\right)} &  \ldots  & {\left(e_{\omega_{i}(N-1)}(1),x_{i}(N-1)\right)}  \\
    \vdots  &  \vdots  &  \ddots  &  \vdots   \\
   {\left(e_{\omega_{i}(0)}(n-1),x_{i}(0)\right)} & {\left(e_{\omega_{i}(1)}(n-1),x_{i}(1)\right)} &  \ldots  & {\left(e_{\omega_{i}(N-1)}(n-1),x_{i}(N-1)\right)}  \\
\end{array}} \right)\\
\!\!&\!\!=\!\!&\!\!\left( {\begin{array}{*{20}{c}}
   {{u_{i}(0)}} & {{u_{i}(1)}} &  \ldots  & {{u_{i}(N-1)}}  \\
   {{u_{i}(N)}} & {{u_{i}(N+1)}} &  \ldots  & {{u_{i}(2N-1)}}  \\
    \vdots  &  \vdots  &  \ddots  &  \vdots   \\
   {{u_{i}((n-1)N)}} & {{u_{i}((n-1)N + 1)}} &  \ldots  & {{u_{i}(nN-1)}}  \\
\end{array}} \right).\nonumber
\end{eqnarray}

By reading the elements in $\bu_i$ row by row, we get an extended FHS $\bu_{i}=(u_{i}(0), u_{i}(1),\ldots, u_{i}(nN-1))$ of period $nN$. Thus, we can obtain the extended FHS set $\mU=\{\bu_i:0\leq i\leq M-1\}$, where the length $n$ of sequence $\be_\delta$ is called the extension factor. For short, we write the extended FHS $\bu_{i}$ as
\[
\bu_{i}=E\left[\be_{\omega_{i}},\bx_{i}\right]=E\left[\left(\be_{\omega_{i}(0)},x_{i}(0)\right),\left(\be_{\omega_{i}(1)},x_{i}(1)\right),\ldots,\left(\be_{\omega_{i}(N-1)},x_{i}(N-1)\right)\right],
\]
where $E$ is the extension operator.

Another extended FHS $\bu_{j}$ can be generated by combining $\be_{\omega^{j}}$ with $\bx_{j}$ as follows:
\[
\bu_{j}=E\left[\be_{\omega_{j}},\bx_{j}\right]=E\left[\left(\be_{\omega_{j}(0)},x_{j}(0)\right),\left(\be_{\omega_{j}(1)},x_{j}(1)\right),\ldots,\left(\be_{\omega_{j}(N-1)},x_{j}(N-1)\right)\right].
\]

Consider its cyclical shift version $L^{(\tau)}(\bu_{j})$, where $L$ is the (left cyclical) shift operator, $\tau=N\tau_{1}+\tau_{2}$, $0\leq\tau_{1}\leq n-1$, $0\leq\tau_{2}\leq N-1$. By the matrix representation, $L^{(\tau)}(\bu_{j})$ could be written as

 \begin{equation} \label{e5}\footnotesize
 %L^{(\tau)}(\bu_{j})=
  \left( {\begin{array}{*{20}{c}}
  {\left(e_{\omega_{j}(\tau_2)}(\tau_1),x_{j}(\tau_2)\right)} &\ldots  & {\left(e_{\omega_{j}(N-1)}(\tau_1),x_{j}(N\!-\!1)\right)} &   \ldots  & {\left(e_{\omega_{j}(\tau_2-1)}(\tau_1),x_{j}(\tau_2\!-\!1)\right)}  \\
       \vdots  &  \vdots  &\vdots  &  \ddots  &  \vdots   \\
    {\left(e_{\omega_{j}(\tau_2)}(n\!-\!1),x_{j}(\tau_2)\right)} &\ldots  & {\left(e_{\omega_{\!j}(N-1)}(n-1),x_{j}(N\!-\!1)\right)} & \ldots  & {\left(e_{\omega_{j}(\tau_2-1)}(n-1),x_{j}(\tau_2\!-\!1)\right)}  \\
   {\left(e_{\omega_{j}(\tau_2)}(0),x_{j}(\tau_2)\right)} &\ldots  & {\left(e_{\omega_{\!j}(N-1)}(0),x_{j}(N\!-\!1)\right)} &   \ldots  & {\left(e_{\omega_{j}(\tau_2-1)}(0),x_{j}(\tau_2\!-\!1)\right)}  \\
    \vdots  &  \vdots  &\vdots  &  \ddots  &  \vdots   \\
   {\left(e_{\omega_{j}(\tau_2)}(\tau_1\!-\!1),x_{j}(\tau_2)\right)} &\ldots  & { \left(e_{\omega_{j}(N-1)}(\tau_1\!-\!1),x_{j}(N\!-\!1)\right)} & \ldots  & {\left(e_{\omega_{j}(\tau_2-1)}(\tau_1\!-\!1),x_{j}(\tau_2\!-\!1)\right)}  \\
\end{array}} \right).
\end{equation}
Obviously, $L^{(\tau)}(\bu_{j})$ is just another extended FHS. Namely, we have
\begin{equation*}\footnotesize
L^{(\tau)}\!(\bu_{j})=E\left[\left(L^{(\tau_1)}(\be_{\omega_{j}(\tau_2)}),x_{j}(\tau_2)\right),\ldots,
\left(L^{(\tau_1)}(\be_{\omega_{j}(N-1)}),x_{j}(N-1)\right),\ldots,
\left(L^{(\tau_1)}(\be_{\omega_{j}(\tau_2-1)}),x_{j}(\tau_2-1)\right)\right].
\end{equation*}

Then, the Hamming correlation function between the extended FHSs $\bu_{i}$ and $\bu_{j}$ at shift $\tau$ becomes the following from (\ref{e4}) and (\ref{e5}),
i.e.,
\begin{eqnarray}\label{e6}
H_{\bu_{i}\bu_{j}}(\tau)
&=&\sum\limits_{t_2=0}^{N-1}\left(\sum\limits_{t_1=0}^{n-1}h\left(e_{\omega_{i}(t_2)}(t_1),(e_{\omega_{j}(t_2+\tau_2)}(t_1+\tau_1))\right)\right)\cdot h\left(x_{i}(t_2),x_{j}(t_2+\tau_2)\right)\nonumber \\
&=&\sum\limits_{t_2=0}^{N-1}\left(H_{\be_{\omega_{i}(t_2)},\be_{\omega_{i}(t_2+\tau_2)}}(\tau_1)\right)\cdot h\left(x_{i}(t_2),x_{j}(t_2+\tau_2)\right).\
\end{eqnarray}

\begin{Lemma}\label{L1}
With the above notation, the nontrivial Hamming correlation between $\bu_{i}$ and $\bu_{j}$ is
 \begin{eqnarray*}
H_{\bu_{i}\bu_{j}}(\tau)\leq\left\{ {\begin{array}{*{20}{c}}
 NH_a(\mE), & i= j,\tau_2=0,\tau_1\neq0,\\
 \lambda H_m(\mE), & i\neq j\  or\  \tau_2\neq0,\\
\end{array}} \right.
\end{eqnarray*}

for any $0\leq i,j\leq M-1, 0\leq \tau_1\leq n-1, 0\leq \tau_2\leq N-1$.

\end{Lemma}
\begin{proof}
Let $\tau=N\tau_{1}+\tau_{2}$, where $0\leq \tau_1\leq n-1, 0\leq \tau_2\leq N-1$. In order to compute $H_{\bu_{i}\bu_{j}}(\tau)$, we divide the problem into two cases.

Case $i)$. $i=j,\tau_2=0$.

If $\tau_1=0$, then
\begin{eqnarray*}
H_{\bu_{i}\bu_{i}}(0)&=&\sum\limits_{t_2=0}^{N-1}n\cdot 1=nN.
\end{eqnarray*}
This is a trivial case.

If $\tau_1\neq0$, then
 \begin{eqnarray*}
H_{\bu_{i}\bu_{i}}(\tau)&=&\sum\limits_{t_2=0}^{N-1}\left(H_{\be_{\omega_{i}(t_2)},\be_{\omega_{i}(t_2)}}(\tau_1)\right)\cdot1\\
&=&\sum\limits_{t_2=0}^{N-1}\left(H_{\be_{\omega_{i}(t_2)},\be_{\omega_{i}(t_2)}}(\tau_1)\right).
\end{eqnarray*}

Therefore, for $i=j,\tau_2=0$
 \begin{eqnarray*}
H_{\bu_{i}\bu_{i}}(\tau)=\left\{ {\begin{array}{*{20}{c}}
nN, & if  \tau_1=0, \\
 \sum\limits_{t_2=0}^{N-1}\left(H_{\be_{\omega_{i}(t_2)},\be_{\omega_{i}(t_2)}}(\tau_1)\right), & otherwise.\\
\end{array}} \right.
\end{eqnarray*}

So, we have $H_{\bu_{i}\bu_{i}}(\tau)\leq NH_a(\mE)$ for $i=j,\tau_2=0,\tau_1\neq0$.

Case $ii)$. $i\neq j$ or $\tau_2\neq0$.

In this case, based on the definition of $\omega$, we have $\omega_{i}(t_2)\neq\omega_{j}(t_2+\tau_2)$ if $ h\left(x_{i}(t_2),x_{j}(t_2+\tau_2)\right)=1$. Thus, we can express the Hamming correlation of $\mS$ by using the maximum Hamming correlation $H_m(\mE)$ of sequence set $\mE$ as
 \begin{eqnarray*}
 H_{\bu_{i}\bu_{j}}(\tau)&=&\sum\limits_{t_2=0}^{N-1}\left(H_{\be_{\omega_{i}(t_2)},\be_{\omega_{j}(t_2+\tau_2)}}(\tau_1)\right)\cdot h\left(x_{i}(t_2),x_{j}(t_2+\tau_2)\right)\\
 &\leq&\sum\limits_{t_2=0}^{N-1}H_c(\mE)\cdot h\left(x_{i}(t_2),x_{j}(t_2+\tau_2)\right)\\
 &\leq&\lambda H_c(\mE).
\end{eqnarray*}

By summarizing the results of Case $i)$ and Case $ii)$, we have
 \begin{eqnarray*}
H_{\bu_{i}\bu_{j}}(\tau)\leq\left\{ {\begin{array}{*{20}{c}}
 NH_a(\mE), & i= j, \tau_2=0, \tau_1\neq0,\\
 \lambda H_c(\mE), & otherwise.\\
\end{array}} \right.
\end{eqnarray*}
\end{proof}

\begin{Lemma}\label{L2}
With the above notation, specially, the nontrivial Hamming correlation between $\bu_{i}$ and $\bu_{j}$ is
 \begin{eqnarray*}
H_{\bu_{i}\bu_{j}}(\tau)\left\{ {\begin{array}{*{20}{c}}
=0, & i= j, \tau_2=0, \tau_1\neq0,\\
\leq \lambda, & otherwise,\\
\end{array}} \right.
\end{eqnarray*}
 if the maximum Hamming autocorrelation and crosscorrelation of set $\mE$ satisfy $H_a(\mE)=0$ and $H_c(\mE)=1$.
\end{Lemma}
For short, $\mE$ is called one-coincidence sequence set if the maximum Hamming autocorrelation and crosscorrelation of set $\mE$ satisfy $H_a(\mE)=0$ and $H_c(\mE)=1$.

Under this framework, two new constructions of FHS set will be presented in the following sections, from which infinitely many new optimal FHSs or FHS sets can be obtained. By choosing the one-coincidence sequence set $\mE$ and the existing optimal FHS sets $\mX$ with different parameters, it is expected that there exist some more classes of optimal FHS sets which can be obtained from the our framework.

\section{Optimal FHS Sets of Length $nN$}
In this section, we present a new design for optimal FHS set from
any known optimal FHS set under the framework in Section 3.

Let $\mF=\{f_{1}, f_{2},\ldots, f_{v}\}$ be a frequency slot set with size $|\mF|=v$. Our procedure of the extension construction is described as follows.

\begin{Construction}\label{con1}\textbf{Construction of optimal FHS sets of length $nN$.}
\begin{enumerate}
\itemindent12pt
\item[Step {\rm1}:] Select an optimal $(N, v, \lambda;M)$ FHS set
\[
\mX=\left\{\bx_i=(x_i(0),x_i(1),\ldots,x_i(N-1)):0 \leq i\leq M-1\right\}.
\]

The maximum number of appearance of frequency slot in $\mX$ is $m(\mX)$.

\item[Step {\rm2}:]Let $q_1=p_1^{a_1}$ for a prime $p_1$ and a positive integer $a_1$. Let $g'$ be a primitive element of $\ZZ^{*}_{p_1}$, so $g=g'$ or $g=g'+p_1$ is the primitive element of $\ZZ^{*}_{q_1}$. Generate a set
\[
\mE=\{\be_j=(0\cdot g^j, 1\cdot g^j, \cdots, (q_1-1)\cdot g^j):0\leq j\leq p_1-2\}.
\]

\item[Step {\rm3}:] If $m(\mX)\leq p_1-1$, let $\omega_i(t_2)$ be a function with $\omega_i(t_2)=\varphi_{f_k}(i,t_2)$, where $\varphi_{f_k}$ is a injective function from $\Phi_k$ to $Z_{l}^*$ for every $k$ with $0\leq k \leq v$, $0\leq i\leq M-1$, $0\leq t_2\leq N-1$, the set $\Phi_k$ is defined as Lemma 2.

\item[Step {\rm4}:]We can construct the desired FHS set $\mS=\{\bs_{i}: 0\leq i\leq M-1\}$,
\[
\bs_i=I\left[\left(\be_{\omega_i(0)},x_i(0)),(\be_{\omega_i(1)},x_i(1)),\cdots,(\be_{\omega_i(N-1)},x_i(N-1)\right)\right].
\]
\end{enumerate}
\end{Construction}

\begin{Lemma}\label{L3}
The set $\mE$ in Construction {\rm\ref{con1}} has the following properties:

$(i)$ $\mE$ is a sequence set with $p_1-1$ sequences of length $q_1$ over $\ZZ_{q_1}$;

$(ii)$ $H_a(\mE)=0$, $H_c(\mE)=1$.
\end{Lemma}

\begin{proof}
It is easily checked that the sequence set $\mE$ in Construciton 1 is a set with $p_1-1$ sequences of length $q_1$ over $\ZZ_{q_1}$. For any $0\leq j_1,j_2\leq p_1-2$, $0\leq t \leq q_1-1$, the Hamming correlation $H_{\be_{j_1}\be_{j_2}}(t)$ between $\be_{j_1}$ and $\be_{j_2}$ is given by
\begin{eqnarray*}
H_{\be_{j_1}\be_{j_2}}(t)&=&\sum\limits_{i=0}^{q_1-1}h\left(\langle ig^{j_1}\rangle_{q_1},\langle(i+t)g^{j_2}\rangle_{q_1}\right).
\end{eqnarray*}

In order to compute $H_{\be_{j_1}\be_{j_2}}(t)$, we divide the problem into two cases.

Case $i)$. $j_1=j_2$.

In this case, since $g$ is a primitive element of $\ZZ^{*}_{q_1}$, we have $ig^{j_1}\neq (i+t)g^{j_1} \mod \ {q_1}$ if $t\neq0$. Then, we can obtain that
\begin{eqnarray*}
H_{\be_{j_1}\be_{j_1}}(t)=\left\{ {\begin{array}{*{20}{c}}
q_1, & t=0,\\
0, & otherwise.\\
\end{array}} \right.
\end{eqnarray*}

Case $ii)$. $j_1\neq j_2$.

Since $g=g'$ or $g'+p_1$ is the primitive element $\in \ZZ^{*}_{q_1}$, where $g'$ is a primitive element $\in \ZZ^{*}_{p_1}$.
Hence, for any $0\leq j_1, j_2\leq p_1-2$, we let $g^\theta\equiv g^{j_1}-g^{j_2} \mod {q_1}$, where $0\leq \theta\leq {\ord(q_1)-1}$.

Then, we have
\begin{eqnarray*}
H_{\be_{j_1}\be_{j_2}}(t)&=&\sum\limits_{i=0}^{q_1-1}h\left(\langle ig^{j_1}\rangle_{q_1},\langle(i+t)g^{j_2}\rangle_{q_1}\right)\\
&=&\sum\limits_{i=0}^{q_1-1}h\left(\langle ig^{\theta}\rangle_{q_1},\langle tg^{j_2}\rangle_{q_1}\right).
\end{eqnarray*}

Thus, we can obtain
\begin{eqnarray*}
H_{\be_{j_1}\be_{j_2}}(t)=1.
\end{eqnarray*}
since $h\left(\langle ig^{\theta}\rangle_{q_1},\langle tg^{j_2}\rangle_{q_1}\right)=1$ only if $i\equiv tg^{(j_2-\theta)} \mod {q_1}$.

By summarizing the results of  Case $i)$ and $ii)$, we have the nontrivial Hamming correlation of $\mE$ is
\begin{eqnarray*}
H_{\be_{j_1}\be_{j_2}}(t)=\left\{ {\begin{array}{*{20}{c}}
0, & j_1=j_2,t\neq0,\\
1, & j_1\neq j_2.\\
\end{array}} \right.
\end{eqnarray*}

\end{proof}

\begin{Theorem}\label{t1}
The FHS set $\mS$ constructed by Construction {\rm\ref{con1}} is an optimal
$(q_1N,q_1v, \lambda; M)$ FHS set over $\ZZ_{q_1}\times\mF$ if $\left\lceil {\frac{{(NM - v)}}{{(NM - 1)}}\frac{N}{v}}\right\rceil=\left\lceil\frac{(q_1NM-v)}{(q_1NM-1)}\frac{N}{v}\right\rceil$.
\end{Theorem}

\begin{proof}
It is easily checked that that the extended FHS set $\mS$ in Construciton 1 is a set with $M$ sequences of length $q_1N$ over $\ZZ_{q_1}\times\mF$.

From Lemma {\rm\ref{L1}},we know that the Hamming correlation of FHS set $\mS$ is based on the correlation of set $\mE$. Thus, from Lemma {\rm\ref{L3}}, we can obtain that $H_m(\mS)=H_m(\mX)=\lambda$.

Then, we verify the optimality of the extended FHS set $\mS$. Since the optimal FHS set $\mX$ satisfies Peng-Fan bound \cite{PF}, we have
$\lambda= \left\lceil {\frac{{(NM - v)}}{{(NM - 1)}}\frac{N}{v}}
\right\rceil$.

According to the Peng-Fan bound \cite{PF}, the FHS set $\mS$ with the sequence length $q_1N$,
the family size $M$ and the frequency slot set of size $q_1v$ over $\ZZ_{q_1}\times\mF$, the optimal maximum Hamming correlation $H'_m$ of $\mS$ should be
\begin{eqnarray*}
{H'_m}\ge \left\lceil\frac{(Mq_1N-v)}{(Mq_1N-1)}\frac{q_1N}{q_1v}\right\rceil=\left\lceil\frac{(Mq_1N-v)}{(Mq_1N-1)}\frac{N}{v}\right\rceil.
\end{eqnarray*}

Thus, we have the FHS set $\mS$ in Construction 1 is optimal if $\left\lceil\frac{(q_1NM-v)}{(q_1NM-1)}\frac{N}{v}\right\rceil= \left\lceil {\frac{{(NM - v)}}{{(NM - 1)}}\frac{N}{v}}\right\rceil$. This constraint is easy to satisfy because $q_1NM>NM$.

\end{proof}

 The above construction remove the constraint requiring that the expansion factor is coprime with the length of original FHSs. Moreover, it is possible to extend the length of $\mS$ by the extension factor
$q_2 = p^{a_2}_2$, where $a_2$ is a positive integer and $p_2$ is a prime with
$p_2 > p_1$. In this way, Construction 1 can be applied recursively infinitely many times. Since Construction 1 can be applied to any existing optimal FHSs or FHS sets, it is expected that there exist some more classes of optimal FHS sets which can be obtained from our construction.

\begin{Corollary}\label{c1}
For positive integers $r$ and $a_1, \ldots, a_r$,
let $n = q_1\cdots q_r$, where $q_i = p^{a_i}_i , 1\leq i \leq r$, $p_1, \ldots, p_r$ are primes with
$p_1 <\ldots< p_r$. If  $m(\mX)\leq p_1-1$, then there exists an optimal
$(nN,nv,\lambda; M)$ FHS set over $\ZZ_{q_r}\times\cdots\times \ZZ_{q_1} \times \mF$.
\end{Corollary}

Construction 1 in the case of $M= 1$ leads to the construction of a new single FHS. Based on Construction 1, it is possible to obtain new optimal FHSs with respect to the Lempel-Greenberger bound.

\begin{Corollary}\label{c2}
Assume that there exists an optimal
$(N, v, \lambda)$-FHS $\bs$ with respect to the Lempel-Greenberger bound, defined over $\mF$. For positive integers $r$ and $a_1, \ldots, a_r$,
let $n = q_1\cdots q_r$, where $q_i = p^{a_i}_i , 1\leq i \leq r$, $p_1, \ldots, p_r$ are primes with
$p_1 <\ldots< p_r$. If $m(\bs) \leq p_1-1$, then there exists an optimal
$(nN,nv,\lambda)$ FHS over $\ZZ_{q_r}\times\cdots\times \ZZ_{q_1} \times \mF$.
\end{Corollary}

\begin{Remark}
By means of our method, Chung's Construction A in \cite{Chung-14} is a special case of Corollary {\rm\ref{c1}} for $\gcd(n,N)=1$ and $\omega_i(t_2)=\sum^{i-1}_{\alpha=0}\sum^{N-1}_{t=0}h(x_\alpha(t),x_i(t_2))+\sum^{t_2-1}_{\beta=0}h(x_i(\beta),x_i(t_2))-1$ for any $0\leq i\leq M-1, 0\leq t_2\leq N-1$.
\end{Remark}

\begin{Example}
\itemindent12pt
\item[Step {\rm1}:]
We select an optimal $(26,7,4;3)$ FHS set $\mX=\{\bx_{0}, \bx_{1},
\bx_{2}\}$ over $\ZZ_7$ , such that
\begin{eqnarray*}
\bx_{0}=\{6,1,5,3,4,4,1,2,3,5,2,0,0,6,0,0,2,5,3,2,1,4,4,3,5,1\}, \\
\bx_{1}=\{6,5,3,1,2,2,5,0,1,3,0,4,4,6,4,4,0,3,1,0,5,2,2,1,3,5\},\\
\bx_{2}=\{6,3,1,5,0,0,3,4,5,1,4,2,2,6,2,2,4,1,5,4,3,0,0,5,1,3\}.
\end{eqnarray*}
So, we have $m(\mX)=12$.

\item[Step {\rm2}:]
Select $q_1=p_1=13>m(\mX)=12$. Choose $g=2$ is the primitive element $\in \ZZ^*_{13}$. Generate a
set $\mE=\{\be_j=(0\cdot g^j,1\cdot g^j,\cdots,12\cdot g^j):0\leq j\leq 11\}$. The set $\mE$ could be written as
\begin{eqnarray*}
\be_{0}&=&\{0,1,2,3,4,5,6,7,8,9,10,11,12\}, \\
\be_{1}&=&\{0,2,4,6,8,10,12,1,3,5,7,9,11\},\\
\vdots\\
\be_{11}&=&\{0,7,1,8,2,9,3,10,4,11,5,12,6\}.
\end{eqnarray*}

\item[Step {\rm3}:] For any $0\leq i\leq2, 0\leq t_2\leq 25$, a special expression of $\omega=\{\omega_0, \omega_{1},\omega_{2}\}$ can be given by
\begin{eqnarray*}
\omega_{0}&=&\{0,0,0,0,0,1,1,0,1,1,1,0,1,1,2,3,2,2,2,3,2,2,3,3,3,3\}, \\
\omega_{1}&=&\{2,4,4,4,4,5,5,4,5,5,5,4,5,3,6,7,6,6,6,7,6,6,7,7,7,7\},\\
\omega_{2}&=&\{4,8,8,8,8,9,9,8,9,9,9,8,9,5,10,11,10,10,10,11,10,10,11,11,11,11\}.
\end{eqnarray*}

\item[Step {\rm4}:]
We can construct the FHS set $\mS=\{\bs_{0}, \bs_{1}, \bs_{2}\}$ by the Construction 1.
\begin{eqnarray*}
\bs_{0}=\{\!\!&\!&\!(0,6),(0,1),(0,5),(0,3),(0,4),(0,4),(0,1),(0,2),(0,3),(0,5), \ldots,\\
            \!&\!&\!(9,2),(9,5),(9,3),(5,2),(9,1),(9,4),(5,4),(5,3),(5,5),(5,1)\}, \\
\bs_{1}=\{\!\!&\!&\!(0,6),(0,5),(0,3),(0,1),(0,2),(0,2),(0,5),(0,0),(0,1),(0,3), \ldots,\\
           \! &\!&\!(1,0),(1,3),(1,1),(2,0),(1,5),(1,2),(2,2),(2,1),(2,3),(2,5)\},\\
\bs_{2}=\{\!\!&\!&\!(0,6),(0,3),(0,1),(0,5),(0,0),(0,0),(0,3),(0,4),(0,5),(0,1), \ldots,\\
            \!&\!&\!(3,4),(3,1),(3,5),(6,4),(3,3),(3,0),(6,0),(6,5),(6,1),(6,3)\}.
\end{eqnarray*}

The FHS $\bs_{0}$ over $\ZZ_{13}\times\ZZ_7$ can be obtained from the $13\times26$ array in (6), that is
\begin{eqnarray}
\ \ \ \ \   \small{\left(\begin{array}{*{70}{c}}
(0,6)&(0,1)&(0,5)&(0,3)&(0,4)&\cdots&&(0,4)&(0,3)&(0,5)&(0,1)\\
(1,6)&(1,1)&(1,5)&(1,3)&(1,4)&\cdots&&(8,4)&(8,3)&(8,5)&(8,1)\\
(2,6)&(2,1)&(2,5)&(2,3)&(2,4)&\cdots&&(3,4)&(3,3)&(3,5)&(3,1)\\
(3,6)&(3,1)&(3,5)&(3,3)&(3,4)&\cdots&&(11,4)&(11,3)&( 11,5)&(11,1)\\
(4,6)&(4,1)&(4,5)&(4,3)&(4,4)&\cdots&&(6,4)&(6,3)&(6,5)&(6,1)\\
(5,6)&(5,1)&(5,5)&(5,3)&(5,4)&\cdots&&(1,4)&(1,3)&(1,5)&( 1,1)\\
(6,6)&(6,1)&(6,5)&(6,3)&(6,4)&\cdots&&(9,4)&(9,3)&( 9,5)&(9,1)\\
(7,6)&(7,1)&(7,5)&(7,3)&(7,4)&\cdots&&(4,4)&(4,3)&(4,5)&(4,1)\\
(8,6)&(8,1)&(8,5)&(8,3)&(8,4)&\cdots&&(12,4)&(12,3)&(12,5)&( 12,1)\\
(9,6)&(9,1)&(9,5)&(9,3)&(9,4)&\cdots&&(7,4)&(7,3)&(7,5)&( 7,1)\\
(10,6)&(10,1)&(10,5)&(10,3)&(10,4)&\cdots&&(2,4)&(2,3)&(2,5)&( 2,1)\\
(11,6)&(11,1)&(11,5)&(11,3)&(11,4)&\cdots&&(10,4)&(10,3)&( 10,5)&(10,1)\\
(12,6)&(12,1)&(12,5)&(12,3)&(12,4)&\cdots&&(5,4)&(5,3)&( 5,5)&(5,1)
\end{array}\right)}. \normalsize
\end{eqnarray}

Then, the maximum Hamming correlation of $\mS$ can be seen in Figure 1.
\begin{figure}\centering
\includegraphics[width=0.65\textwidth,bb=160 650 460 790]{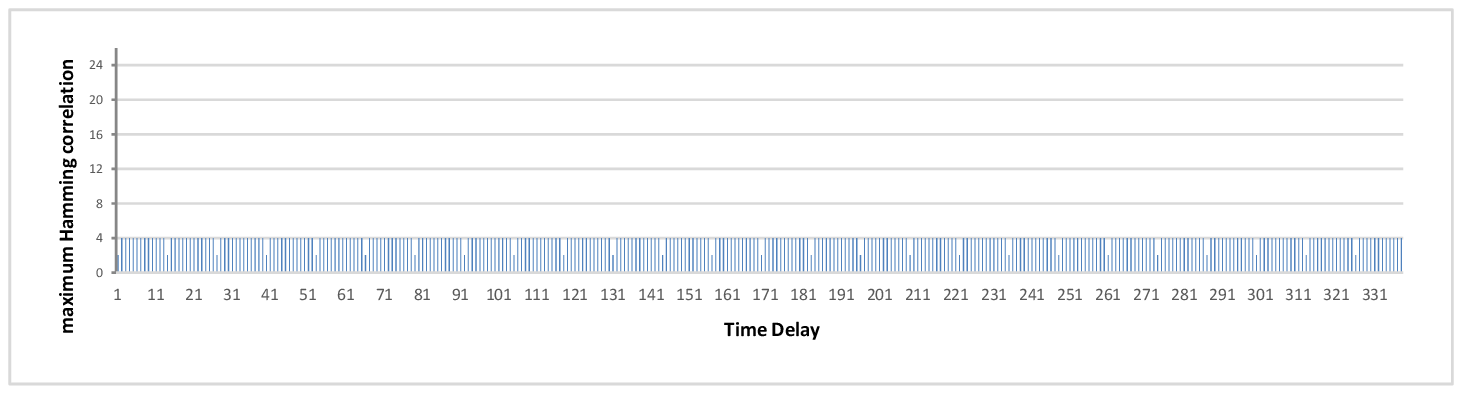}
\caption{The maximum Hamming correlations of $\mS$ in Example 1}
\end{figure}

The maximum Hamming correlation of FHS set $\mS$ in Example 1 is $\lambda=4$. Thus, it can be verified that $S$ is an optimal $(338, 91, 4; 3)$  FHS
set. In the similar way, $\mX$ can be extended to an optimal FHS set with other parameters by choosing different $n = q_1\cdots q_r$, and the extension factor $n$ should not co-prime to the length $N=26$ of $\mX$.
\end{Example}

\section{Optimal FHS Sets of Length $dN$}
In this section, we present a new design for optimal FHS set with new parameters under the framework in Section 3. Our procedure of the extension construction is described as follows.

\begin{Construction}\label{con2}\textbf{Construction of optimal FHS sets of length $dN$.}
\begin{enumerate}
\itemindent12pt
\item[Step {\rm1}:] Select an optimal $(N, v, \lambda;M)$ FHS set
\[
\mX=\left\{\bx_i=(x_i(0),x_i(1),\ldots,x_i(N-1)):0 \leq i \leq M-1\right\}.
\]
The maximum number of appearance of frequency slot in $\mX$ is $m(\mX)$.

\item[Step {\rm2}:] Let $q=p^{a}$ for a prime $p$ and a positive integer $a$.  Let $g'$ be a primitive element of $\in \ZZ^{*}_{p}$, so $g=g'$ or $g=g'+p$ is the primitive element of $\ZZ^{*}_{q}$, $d=\ord(g)=(p-1)p^{(a-1)}$. Generate a set
\[
\mE=\{\be_j=(g^0+j, g^1+j, \cdots, g^{d-1}+j):0\leq j\leq p-1\}.
\]

\item[Step {\rm3}:] If $m(\mX)\leq p$, $\omega_i(t_2)$ be a function with $\omega_i(t_2)=\varphi_{f_k}(i,t_2)$, where $\varphi_{f_k}$ is a injective function from $\Phi_k$ to $Z_{l}^*$ for every $k$ with $0\leq k \leq v$, $0\leq i\leq M-1$, $0\leq t_2\leq N-1$, the set $\Phi_k$ is defined as Lemma 2.

\item[Step {\rm4}:]We can construct the desired FHS set $\mS=\{\bs^{i}: 0\leq i\leq M-1\}$,
\[
\bs^i=I[(\be_{\omega_i(0)},a_i(0)),(\be_{\omega_i(1)},a_i(1)),\cdots,(\be_{\omega_i(N-1)},a_i(N-1))].
\]
\end{enumerate}
\end{Construction}

\begin{Lemma}\label{L4}
The set $\mE$ in Construction {\rm\ref{con2}} has the following properties:
\\$(i)$ $\mE$ is a sequence set with $p$ sequences of length $d=(p-1)p^{(a-1)}$ over $\ZZ_{q}$;
\\$(ii)$ $H_a(\mE)=0$, $H_c(\mE)=1$;
\end{Lemma}

\begin{proof}
It is easily checked that the set $\mE$ in Construciton 2 is with $p$ sequences of length $d=(p-1)p^{(a-1)}$ over $\ZZ_{q}$. For any $0\leq t \leq d-1$, $0\leq j_1,j_2\leq p-1$, the Hamming correlation $H_{\be_{j_1}\be_{j_2}}(t)$ between $\be_{j_1}$ and $\be_{j_2}$ is given by
\begin{eqnarray*}
H_{\be_{j_1}\be_{j_2}}(t)&=&\sum\limits_{i=0}^{d-1}h\left(\langle g^i+{j_1}\rangle_{q},\langle g^{i+t}+{j_2}\rangle_{q}\right).
\end{eqnarray*}

In order to compute $H_{\be_{j_1}\be_{j_2}}(t)$, we divide the problem into two cases.

Case $i)$. $j_1=j_2$.

In this case, we can easily obtain that
\begin{eqnarray*}
H_{\be_{j_1}\be_{j_1}}(t)=\left\{ {\begin{array}{*{20}{c}}
d, & t=0,\\
0, & otherwise.\\
\end{array}} \right.
\end{eqnarray*}

Case $ii)$. $j_1\neq j_2$.

Since $g=g'$ or $g=g'+p_1$ is the primitive element $\in \ZZ^{*}_{q_1}$, where $g'$ is a primitive element $\in \ZZ^{*}_{p_1}$.
For any $0\leq i\leq d-1$, $0\leq t\leq d-1$, we have
\begin{eqnarray}
\langle g^{i}-g^{i+t}\rangle_q=\left\{ {\begin{array}{*{20}{c}}
g^\rho, & t\neq p-1,\\
\sigma p, & t=p-1,\\
\end{array}} \right.
\end{eqnarray}
where $0\leq \rho\leq {d-1},1\leq\sigma <p^{a-1}$.

Thus, we have
\begin{eqnarray*}
H_{\be_{j_1}\be_{j_2}}(t)&=&\sum\limits_{i=0}^{d-1}h\left(\langle g^i+{j_1}\rangle_{q},\langle g^{i+t}+{j_2}\rangle_{q}\right)\\
&=&\sum\limits_{i=0}^{d-1}h\left(\langle g^{i}-g^{i+t}\rangle_{q},\langle j_2-j_1\rangle_{q}\right).
\end{eqnarray*}

Without loss of generality, we can assume $0\leq j_1<j_2\leq p-1$, so we have $0\leq j_2-j_1\leq p-1$.

Hence, we can obtain from equation (7) that
\begin{eqnarray*}
H_{\be_{j_1}\be_{j_2}}(t)\left\{ {\begin{array}{*{20}{c}}
1, & t\neq p-1,\\
0, & t=p-1.\\
\end{array}} \right.
\end{eqnarray*}

By summarizing the results of  Case $i)$ and Case $ii)$, we have the nontrivial Hamming correlation of $\mE$ is
\begin{eqnarray*}
H_{\be_{j_1}\be_{j_2}}(t)=\left\{ {\begin{array}{*{20}{c}}
0, & j_1=j_2,t\neq0,\\
0, & j_1\neq j_2,t=p-1,\\
1, & j_1\neq j_2,t\neq p-1.\\
\end{array}} \right.
\end{eqnarray*}

\end{proof}

\begin{Theorem}\label{t2}
The FHS set $\mS$ constructed by Construction {\rm\ref{con2}} is an optimal
$(dN,qv,\lambda; M)$ FHS set over $\ZZ_{q}\times\mF$ if $\left\lceil\frac{(dNM-v)}{(qNM-1)}\frac{dN}{qv}\right\rceil= \left\lceil {\frac{{(NM - v)}}{{(NM - 1)}}\frac{N}{v}}\right\rceil$.
\end{Theorem}
\begin{proof}
It is easily checked that the extended FHS set $\mS$ in Construciton 2 is a set with $M$ sequences of length $dN$ over $\ZZ_{q}\times\mF$, where $d=(p-1)p^{(a-1)}$.

From Lemma {\rm\ref{L1}}, the Hamming correlation of FHS set $\mS$ is based on the correlation of set $\mE$. Thus, from Lemma {\rm\ref{L3}}, we can obtain that $H_m(\mS)=H_m(\mX)=\lambda$.
Since the optimal FHS set $\mX$ satisfies Peng-Fan bound \cite{PF}, we have $\lambda= \left\lceil {\frac{{(NM - v)}}{{(NM - 1)}}\frac{N}{v}}
\right\rceil$.

Then, we verify the optimality of the extended FHS set $\mS$. According to the Peng-Fan bound \cite{PF}, the FHS set $\mS$ with length $dN$,
the family size $M$ and the frequency slot set size $qv$ over $\ZZ_{q}\times\mF$, the optimal maximum Hamming correlation $H'_m$ of $\mS$ is
\begin{eqnarray*}
{H'_m}\ge \left\lceil\frac{(dNM-v)}{(qNM-1)}\frac{dN}{qv}\right\rceil.
\end{eqnarray*}

Thus, we have the FHS set $\mS$ in Construction 2 is optimal if $\left\lceil\frac{(dNM-v)}{(qNM-1)}\frac{dN}{qv}\right\rceil= \left\lceil {\frac{{(NM - v)}}{{(NM - 1)}}\frac{N}{v}}\right\rceil$.

\end{proof}

In particular, by generating the set $\mE$ in the finite field $\FF_{q}$, it is possible to construct an extended FHS set of length $dN$ with different extension factor $d$.

\begin{Corollary}\label{c3}
 Let $g$ be a primitive element of $\FF_{q}$, then we have $d=\ord(g)=q-1$. Then, if $m(\mX)\leq q$, the FHS set $\mS$ constructed by Construction {\rm\ref{c2}} is an optimal $((q-1)N, qv, \lambda; M)$ FHS set over $\FF_{q}\times\mF$.
\end{Corollary}

\begin{Remark}
By means of our method, Chung's Construction B in \cite{Chung-14} is a special case of Corollary {\rm\ref{c3}} for $\gcd(q-1,N)=1$  and $\omega_i(t_2)=\sum^{i-1}_{\alpha=0}\sum^{N-1}_{t=0}h(x_\alpha(t),x_i(t_2))+\sum^{t_2-1}_{\beta=0}h(x_i(\beta),x_i(t_2))-1$ for any $0\leq i\leq M-1, 0\leq t_2\leq N-1$.
\end{Remark}

Therefore, it is possible to extend the length of $\mS$ in Construction 2 by the extension factor
$n=q_1\cdots q_r$, if $m(\mX) \leq p_1-1$, where $q_i = p^{a_i}_i , 1\leq i \leq r$, $p_1, \ldots, p_r$ are primes with $p_1 <\ldots< p_r$. In this way, Construction 2 can be applied recursively infinitely many times.

\begin{Corollary}\label{c4}
For positive integers $r$ and $a, a_1, \ldots, a_r$, let $n= q_1\cdots q_r$, where $q = p^{a}$, $q_i = p^{a_i}_i , 1\leq i \leq r$, $p, p_1, \ldots, p_r$ are primes with $p_1 <\ldots< p_r$. If $m(\mX) \leq \min\{p_1-1,p\}$, there exists an optimal
$(ndN,nqv,\lambda; M)$ FHS set over $\ZZ_{q_r}\times\cdots\times \ZZ_{q_1}\times \ZZ_{q} \times \mF$, $d=(p-1)p^{(a-1)}$.
\end{Corollary}

\begin{Corollary}\label{c5}
 Let $g$ be a primitive element of $\FF_{q}$, then the extension factor is $d=\ord$$(g)=q-1$. Thus, if $m(\mX) \leq \min\{p_1-1,q\}$, the FHS set $\mS$ constructed by Construction {\rm\ref{c2}} is an optimal $(n(q-1)N, nqv, \lambda; M)$ FHS set over $\ZZ_{q_r}\times\cdots\times \ZZ_{q_1}\times \FF_{q} \times \mF$.
\end{Corollary}

Some new optimal FHS sets obtained by recursively applying Construction 2 are listed in Table 2. Since Construction 2 can be applied to any existing optimal FHSs or FHS sets, it is expected that there exist some more classes of optimal FHS sets which can be obtained from our construction, but are not listed there.

\begin{table}[tbp]\label{t2}\small
\newcommand{\tabincell}[2]{\begin{tabular}{@{}#1@{}}#2\end{tabular}}
\centering  % 表居中
\caption{The parameters of some new optimal FHS Sets obtained by combining Construction 1 and Construction 2.}

\begin{tabular}{c|c|c}  %p{3cm} {lccc} 表示各列元素对齐方式，left-l,right-r,center-c
\hline
 New FHS sets &  Constraints & References\\ \hline  % \hline 在此行下面画一横线
$\left(nd(\alpha^m-1),nq\alpha^l,\alpha^{m-l}; \alpha^l\right)$ &  \tabincell{l}{$1\leq l \leq m$,\\$d=(p-1)p^{a-1}$, $\alpha^m< \min\{p,\lpf(n)-1\}$. \\or $d=p^a-1$, $\alpha^m< \min\{q,\lpf(n)-1\}$.} & \cite{LG}\\ \hline
$\left(nd\alpha^2,nq\alpha,\alpha; \alpha\right)$ &  \tabincell{l}{$d=(p-1)p^{a-1}$, $\alpha^2\leq \min\{p,\lpf(n)-1\}$. \\or $d=p^a-1$, $\alpha^2\leq \min\{q,\lpf(n)-1\}$.}& \cite{KUMAR,XU-16}\\ \hline
$\left(nd\alpha,nq(e+1),g; e\right)$ &  \tabincell{l}{$\alpha=eg+1, 2\leq g<e$,\\$d=(p-1)p^{a-1}$, $\alpha< \min\{p,\lpf(n)-1\}$. \\or $d=p^a-1$, $\alpha< \min\{q,\lpf(n)-1\}$.}& \cite{CHU}\\ \hline
$\left(nd(r-1),nq(e+1),g; e\right)$ &  \tabincell{l}{$r=eg+1, r-1\leq e(e+1)$,\\$d=(p-1)p^{a-1}$, $r< \min\{p,\lpf(n)-1\}$. \\or $d=p^a-1$, $r< \min\{q,\lpf(n)-1\}$.}& \cite{DING-08,HAN}\\ \hline
$\left(ndr^m-1,nqr^l,r^{m-l};r^l\right)$ &  \tabincell{l}{$1\leq l \leq m$,\\$d=(p-1)p^{a-1}$, $r^m< \min\{p,\lpf(n)-1\}$. \\or $d=p^a-1$, $r^m< \min\{q,\lpf(n)-1\}$.}&\cite{ZHOU-11}\\ \hline
$\left(\frac{ndr^m-1}{h},nqr^l,\frac{r^{m-l}-1}{h};h \right)$ &  \tabincell{l}{$\gcd(h,m)=1$, $h|r-1$, $1\leq l \leq m$,\\$d=(p-1)p^{a-1}$, $r^{m-l}< \min\{p,\lpf(n)-1\}$. \\or $d=p^a-1$, $r^{m-l}< \min\{q,\lpf(n)-1\}$.}&\cite{ZHOU-11}\\ \hline

\end{tabular}
\\
 In  Table 2, $\lpf(n)$ denotes the least prime factor of $n$, $q=p^a$ for a prime $p$, $r$ is a prime power for a prime $\alpha$.
\end{table}

Construction 2 in the case of $M= 1$ leads to the construction of a new single FHS. Based on Construction 2, it is possible to obtain new optimal FHSs with respect to the Lempel-Greenberger bound.

\begin{Corollary}\label{c6}
Assume that there exists an optimal
$(N, v, \lambda)$-FHS $\bs$ with respect to the Lempel-Greenberger bound, defined over $\mF$. For positive integers $r$ and $a, a_1, \ldots, a_r$,
$n= q_1\cdots q_r$, where $q=p^a, q_i = p^{a_i}_i , 1\leq i \leq r$, $p, p_1, \ldots, p_r$ are primes with
$p_1 <\ldots< p_r$. Then, there exists an optimal
$(ndN,nqv,\lambda)$-FHS over $\ZZ_{q_r}\times\cdots\times \ZZ_{q_1} \times \ZZ_{q}\times \mF$, if $m(\mX) \leq \min\{p_1-1,q\}$, $d=p^a-1$, or $m(\mX) \leq \min\{p_1-1,p\}$, $d=(p-1)p^{(a-1)}$.
\end{Corollary}

\begin{Example}
\itemindent12pt
\item[Step {\rm1}:]
We select an optimal $(8, 3, 3; 3)$ FHS set $\mY=\{\by_{0}, \by_{1},
\by_{2}\}$ over $\ZZ_3$, such that
\begin{eqnarray*}
\by_{0}=\{0,2,2,1,0,1,1,2\}, \\
\by_{1}=\{1,0,0,2,1,2,2,0\},\\
\by_{2}=\{2,1,1,0,2,0,0,1\}.
\end{eqnarray*}
So, we have $m(\mY)=8$.

\item[Step {\rm2}:]
Select $a=2$, and $q=p^{2}=121$, $p=11>m(\mY)=8$. Choose $g=2$ is the primitive element of $\ZZ^*_{p}$ and $\ZZ^*_{q}$.

Generate a sequence set $\mE=\{\be_j=(g^0+j, g^1+j, \cdots, g^{d-1}+j):0\leq j\leq 10\}$, which can be written as
\begin{eqnarray*}
&&\be_{0}=\{1,2,4,8,16,32,64,7,14,28,56,112,103,85,49,98,75\cdots\}, \\
&&\be_{1}=\{2,3,5,9,17,33,65,8,15,29,57,113,104,86,50,99,76\cdots\},\\
&&\vdots\\
&&\be_{10}=\{11,12,14,18,26,42,74,17,24,38,66,1,113,95,59,108,85\cdots\}.
\end{eqnarray*}
\item[Step {\rm3}:]
For any $0\leq i\leq2$, $0\leq t_2\leq 7$, a special expression of $\omega=\{\omega_0, \omega_{1},\omega_{2}\}$ can be given by
\begin{eqnarray*}
&&\omega_{0}=\{0,0,1,0,1,1,2,2\}, \\
&&\omega_{1}=\{3,2,3,3,4,4,5,4\},\\
&&\omega_{2}=\{6,5,6,5,7,6,7,7\}.
\end{eqnarray*}
\item[Step {\rm4}:]
Then we can construct the FHS set $\mS=\{\bs_{0}, \bs_{1}, \bs_{2}\}$ by the Construction 2.
\begin{eqnarray*}
\bs_{0}&=&\{(1,0),(1,2),(2,2),(1,1),(2,0),(2,1),(3,1),(3,2),(2,0),(2,2),(3,2),(2,1), \ldots,\\
            &&(92,1),(93,1),(93,2),(61,0),(61,2),(62,2),(61,1),(62,0),(62,1),(63,1),(63,2)\}, \\
\bs_{1}&=&\{(4,1),(3,0),(4,0),(4,2),(5,1),(5,2),(6,2),(5,0),(5,1),(4,0),(5,0),(5,2),\ldots,\\
           &&(95,2),(96,2),(95,0),(64,1),(63,0),(64,0),(64,2),(65,1),(65,2),(66,2),(65,0)\},\\
\bs_{2}&=&\{(7,2),(6,1),(7,1),(6,0),(8,2),(7,0),(8,0),(8,1),(8,2),(7,1),(8,1),(7,0), \ldots,\\
            &&(97,0),(98,0),(98,1),(67,2),(66,1),(67,1),(66,0),(68,2),(67,0),(68,0),(68,1)\}.
\end{eqnarray*}
The FHS $\bs_{0}$ over $\ZZ_{121}\times\ZZ_3$ can be obtained from the $110\times 8$ array in $(8)$, that is
\begin{eqnarray}
\left( \small{\begin{array}{*{10}{c}}
(1,0)&(1,2)&(2,2)&(1,1)&(2,0)&(2,1)&(3,1)&(3,2)\\
(2,0)&(2,2)&(3,2)&(2,1)&(3,0)&(3,1)&(4,1)&(4,2)\\
(4,0)&(4,2)&(5,2)&(4,1)&(5,0)&(5,1)&(6,1)&(6,2)\\

\vdots&\vdots&\vdots&\vdots&\vdots&\vdots&\vdots&\vdots\\

(91,0)&(91,2)&(92,2)&(91,1)&(92,0)&(92,1)&(93,1)&(93,2)\\
(61,0)&(61,2)&(62,2)&(61,1)&(62,0)&(62,1)&(63,1)&(63,2)
\end{array}} \normalsize\right).
\end{eqnarray}
Then, the maximum Hamming correlation of $\mS$ can be seen in Figure 2.

\begin{figure}\centering
\includegraphics[width=\textwidth,bb=90 619 535 733]{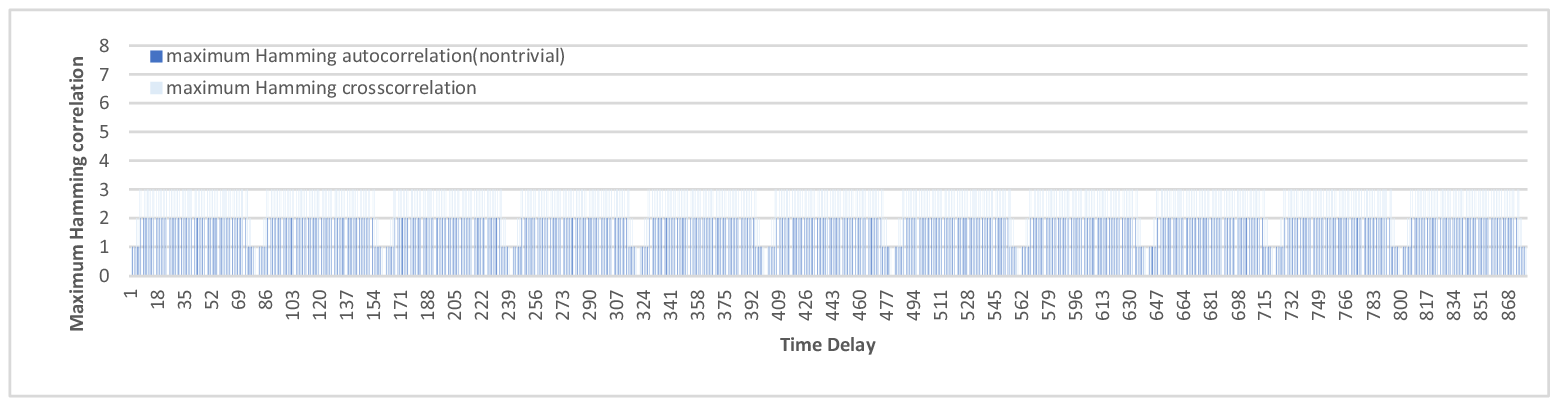}
\caption{The maximum Hamming correlations of $\mS$ in Example 2}
\end{figure}
The maximum Hamming correlation in Example 2 is  $\lambda=3$. Thus, It can be verified that $\mS$ is an optimal $(880,363,3;3)$ FHS
set. In a similar way, original FHS set $\mY$ can be extended to an optimal FHS set with new parameters, and the extension factor $d$ should not co-prime to the length $N=8$ of $\mY$.
\end{Example}

\section{Conclusions}

In this paper, we present a framework of constructing optimal FHS sets based on the designated direct product, which increase the length and alphabet size of the original FHS set, but preserve its maximum Hamming correlation. The framework generalize the previous extension constructions of optimal FHS sets.
Under the framework, we obtain infinitely many new optimal FHS sets by combining a family of sequences that are newly constructed in this paper with some known optimal FHS sets. Compared with the previous extension methods in [25,26], our construction remove the constraint requiring that the expansion factor is co-prime with the length of original FHSs and get new parameters, as shown in Table \ref{tab:1}. As a result, we have a great flexibility of choosing parameters of FHS sets for a given frequency hopping spread spectrum system.

\end{document}